\documentclass[conference]{IEEEtran}
\usepackage{blindtext, graphicx}
\usepackage{amsfonts,amsmath,amssymb,amsthm}

\newtheorem{lemma}{Lemma}

\begin{document}
%
\title{A Stronger Soft-Covering Lemma and Applications}

\author{
\IEEEauthorblockN{Paul Cuff}
\IEEEauthorblockA{Princeton University}
}

\maketitle

\begin{abstract}
Wyner's soft-covering lemma is a valuable tool for achievability proofs of information theoretic security, resolvability, channel synthesis, and source coding.  The result herein sharpens the claim of soft-covering by moving away from an expected value analysis.  Instead, a random codebook is shown to achieve the soft-covering phenomenon with high probability.  The probability of failure is doubly-exponentially small in the block-length, enabling more powerful applications through the union bound.
\end{abstract}

\section{Claim}

Given a channel $Q_{V|U}$ and an input distribution $Q_U$, let the output distribution be $Q_V$.  Also, let the $n$-fold memoryless extensions of these be denoted $Q_{V^n|U^n}$, $Q_{U^n}$, and $Q_{V^n}$.

Wyner's soft-covering lemma \cite[Theorem~6.3]{wyner-common-info} says that the distribution induced by selecting a $U^n$ sequence at random from an appropriately chosen set and passing this sequence through the memoryless channel $Q_{V^n|U^n}$ will be a good approximation of $Q_{V^n}$ in the limit of large $n$ as long as the set is of size greater than $2^{nR}$ where $R > I(U;V)$.  In fact, the set can be chosen quite carelessly---by random codebook construction, drawing each sequence independently from the distribution $Q_{U^n}$.

The soft-covering lemmas in the literature use a distance metric on distributions (commonly total variation or relative entropy) and claim that the distance between the induced distribution $P_{V^n}$ and the desired distribution $Q_{V^n}$ vanishes in expectation over the random selection of the set.\footnote{Many of the theorems only claim existence of a good codebook, but all of the proofs use expected value to establish existence.}  In the literature, \cite{han-verdu} studies the fundamental limits of soft-covering as ``resolvability,'' \cite{hayashi06} provides rates of exponential convergence, \cite{cuff13} improves the exponents and extends the framework, \cite{ahlswede-winter02} and \cite[Chapter~16]{wilde-text} refer to soft-covering simply as ``covering'' in the quantum context, \cite{winter05} refers to it as a ``sampling lemma'' and points out that it holds for the stronger metric of relative entropy, and \cite{hou-kramer14} gives a recent direct proof of the relative entropy result.

Here we give a stronger claim.  With high probability with respect to the set construction, the distance will vanish exponentially quickly with the block-length $n$.  The negligible probability of the random set not producing this desired result is doubly-exponentially small.

Let us define precisely the induced distribution.  Let ${\cal C} = \{u^n(m)\}_{m=1}^M$ be the set of sequences, which will be referred to as the codebook.  The size of the codebook is $M = 2^{nR}$.  Then the induced distribution is:
\begin{align}
    P_{V^n|{\cal C}} &= 2^{-nR} \sum_{u^n(m) \in {\cal C}} Q_{V^n|U^n=u^n(m)}.
\end{align}

\begin{lemma}
    For any $Q_{U}$, $Q_{V|U}$, and $R > I(U;V)$, where $V$ has a finite support ${\cal V}$, there exists a $\gamma_1 > 0$ and a $\gamma_2 > 0$ such that for $n$ large enough
    \begin{align}
        {\mathbf P} \left( d(P_{V^n|{\cal C}}, Q_{V^n}) > e^{-\gamma_1 n} \right) & \leq e^{- e^{\gamma_2 n}},
    \end{align}
    where $d(\cdot,\cdot)$ is the relative entropy.
\end{lemma}

\begin{proof}
    We state the proof in terms of arbitrary distributions (not necessarily discrete).  When needed, we will specialize to the case that ${\cal V}$ is finite.
    
    Let the Radon-Nikodym derivative between the induced and desired distributions be denoted as
    \begin{align}
        D_{\cal C}(v^n) &\triangleq \frac{d P_{V^n|{\cal C}}}{d Q_{V^n}}(v^n).
    \end{align}
    In the discrete case, this is just a ratio of probability mass functions.
    
    Notice that the relative entropy of interest, which is a function of the codebook ${\cal C}$, is given by
    \begin{align}
        d(P_{V^n|{\cal C}},Q_{V^n}) &= \int d P_{V^n|{\cal C}} \log D_{\cal C}.
    \end{align}
    
    Define the jointly-typical set over $u$ and $v$ sequences by
    \begin{align}
        {\cal A}_{\epsilon} &\triangleq \left\{ (u^n, v^n) : \frac{1}{n} \log \frac{d Q_{V^n|U^n=u^n}}{d Q_{V^n}} (v^n) \leq I_Q(U;V) + \epsilon \right\}.
    \end{align}
    
    We split $P_{V^n|{\cal C}}$ into two parts, making use of the indicator function denoted by $\mathbf{1}$.  Let $\epsilon>0$ be arbitrary, to be determined later.
    \begin{align}
        P_{{\cal C}, 1} &\triangleq 2^{-nR} \sum_{u^n(m) \in {\cal C}} Q_{V^n|U^n=u^n(m)} \mathbf{1}_{(V^n,u^n(m)) \in {\cal A}_{\epsilon}}, \\
        P_{{\cal C}, 2} &\triangleq 2^{-nR} \sum_{u^n(m) \in {\cal C}} Q_{V^n|U^n=u^n(m)} \mathbf{1}_{(V^n,u^n(m)) \notin {\cal A}_{\epsilon}}.
    \end{align}
    The measures $P_{{\cal C}, 1}$ and $P_{{\cal C}, 2}$ on the space ${\cal V}^n$ are not probability measures, but $P_{{\cal C}, 1} + P_{{\cal C}, 2} = P_{V^n|{\cal C}}$ for each codebook ${\cal C}$.
    
    Let us also split $D_{\cal C}$ into two parts:
    \begin{align}
        D_{{\cal C}, 1}(v^n) &\triangleq \frac{d P_{{\cal C}, 1}}{d Q_{V^n}}(v^n), \\
        D_{{\cal C}, 2}(v^n) &\triangleq \frac{d P_{{\cal C}, 2}}{d Q_{V^n}}(v^n).
    \end{align}
    
    By Jensen's inequality (or the data processing inequality) we can upper bound the relative entropy of interest:
    \begin{align}
        d(P_{V^n|{\cal C}},Q_{V^n}) &\leq h\left( \int d P_{{\cal C},1} \right) \ldots \nonumber \\
        & \quad + \int d P_{{\cal C},1} \log D_{{\cal C},1} + \int d P_{{\cal C},2} \log D_{{\cal C},2}, \label{expanded divergence bound}
    \end{align}
    where $h(\cdot)$ is the binary entropy function.
    
    Notice that $P_{{\cal C}, 1}$ will usually contain almost all of the probability.  That is, denoting the complement of ${\cal A}_{\epsilon}$ as $\overline{{\cal A}_{\epsilon}}$,
    \begin{align}
        \int d P_{{\cal C}, 2} &= 1 - \int d P_{{\cal C}, 1} \\
        &= 2^{-nR} \sum_{u^n(m) \in {\cal C}} \mathbf{P}_Q \left( \overline{{\cal A}_{\epsilon}} \; | \; U^n = u^n(m, {\cal C}) \right).
    \end{align}
    This is an average of exponentially many i.i.d. random variables bounded between 0 and 1.  Furthermore, the expected value of each one is the exponentially small probability of correlated sequences being atypical:
    \begin{align}
        \mathbf{E} \; \mathbf{P}_Q \left( \overline{{\cal A}_{\epsilon}} \; | \; U^n = u^n(m, {\cal C}) \right) &= \mathbf{P}_Q \left( \overline{{\cal A}_{\epsilon}} \right) \\
        &\leq 2^{-\beta n}, \label{atypical probability bound}
    \end{align}
    where
    \begin{align}
        \beta &= \max_{\alpha > 1} (\alpha - 1) \left( I_Q(U;V) + \epsilon - d_{\alpha}(Q_{U,V}, Q_U Q_V) \right),
    \end{align}
    where $d_{\alpha}(\cdot,\cdot)$ is the R\'{e}nyi divergence of order $\alpha$.  We use units of bits for mutual information and R\'{e}nyi divergence to coincide with the base two expression of rate.
    
    Therefore, the Chernoff bound assures that $\int d P_{{\cal C}, 2}$ is exponentially small.  That is, for any $\beta_1 < \beta$,
    \begin{align}
        {\mathbf P} \left( \int d P_{{\cal C}, 2} \geq 2 \cdot 2^{-\beta_1 n} \right) &\leq e^{-\frac{1}{3} 2^{n( R - \beta_1)}}.
    \end{align}
    
    Similarly, $D_{{\cal C}, 1}$ is an average of exponentially many i.i.d. and uniformly bounded functions, each one determined by one sequence in the codebook:
    \begin{align}
        D_{{\cal C}, 1}(v^n) &= 2^{-nR} \sum_{u^n(m) \in {\cal C}} \frac{d Q_{V^n|U^n=u^n(m)}}{d Q_{V^n}} (v^n) \mathbf{1}_{(v^n,u^n(m)) \in {\cal A}_{\epsilon}}
    \end{align}
    For every term in the average, the indicator function bounds the value to be between $0$ and $2^{nI(U;V) + n \epsilon}$.
    The expected value of each term with respect to the codebook is bounded above by one, which is observed by removing the indicator function.
    Therefore, the Chernoff bound assures that $D_{{\cal C},1}$ is exponentially close to one for every $v^n$.  For any $\beta_2$:
    \begin{align}
        {\mathbf P} \left( D_{{\cal C},1}(v^n) \geq 1 + 2^{-\beta_2 n} \right) &\leq e^{-\frac{1}{3} 2^{n( R - I_Q(U;V) - \epsilon - 2 \beta_2 )}} \quad \forall v^n. \label{typical set bound}
    \end{align}
    This use of the Chernoff bound has been used before for a soft-covering lemma in the proof of Lemma~9 of \cite{ahlswede-winter02}.
    
    At this point we will use the fact that ${\cal V}$ is a finite set to obtain two bounds.  First,
    \begin{align}
        D_{{\cal C},2}(v^n) &\leq \left( \max_{v \in {\cal V}} \frac{1}{Q_V(v)} \right)^n \quad \forall v^n \in {\cal V}^n \quad w. p. 1.
    \end{align}
    Notice that the maximum is only over the support of $V$, which makes this bound finite.  The reason this restriction is possible is because with probability one a conditional distribution is absolutely continuous with respect to its associated marginal distribution.
    
    Next we use the union bound applied to \eqref{atypical probability bound} and \eqref{typical set bound}, taking advantage of the fact that the space ${\cal V}^n$ is only exponentially large.  Let ${\cal S}$ be the set of codebooks such that all of the following are true:
    \begin{align}
    \int d P_{{\cal C},2} &< 2 \cdot 2^{-\beta_1 n}, \\
    D_{{\cal C}, 1}(v^n) &< 1 + 2^{-\beta_2 n} \quad \forall v^n \in {\cal V}^n, \\
    D_{{\cal C}, 2}(v^n) &< \left( \max_{v \in {\cal V}} \frac{1}{Q_V(v)} \right)^n \quad \forall v^n \in {\cal V}^n.
    \end{align}
    We see that the probability of not being in ${\cal S}$ is doubly exponentially small:
    \begin{align}
    \mathbf{P}({\cal C} \notin {\cal S}) &\leq e^{-\frac{1}{3} 2^{n (R - \beta_1)}} + |{\cal V}|^n e^{-\frac{1}{3} 2^{n(R - I_Q(U;V) - \epsilon - 2 \beta_2)}}.
    \end{align}
    
    What remains is to show that for every codebook in ${\cal S}$, the relative entropy is exponentially small.  We begin from \eqref{expanded divergence bound}.
    Since
    \begin{align}
    h(x) &\leq x \log \frac{e}{x},
    \end{align}
    we have
    \begin{align}
    h\left( \int d P_{{\cal C},1} \right) &= h\left( \int d P_{{\cal C},2} \right) \\
    &\leq 2 \cdot 2^{-\beta_1 n} (\beta_1 n \log 2 + \log e - \log 2).
    \end{align}
    Furthermore,
    \begin{align}
    \int d P_{{\cal C},1} \log D_{{\cal C},1} &\leq \int d P_{{\cal C},1} \log (1 + 2^{-\beta_2 n}) \\
    &\leq \log (1 + 2^{-\beta_2 n}) \\
    &\leq 2^{-\beta_2 n} \log e.
    \end{align}
    Finally,
    \begin{align}
    \int d P_{{\cal C},2} \log D_{{\cal C},2} &\leq \int d P_{{\cal C},2} \log \left( \max_{v \in {\cal V}} \frac{1}{Q_V(v)} \right)^n \\
    &\leq n \log \left( \max_{v \in {\cal V}} \frac{1}{Q_V(v)} \right) \int d P_{{\cal C},2} \\
    &\leq n \log \left( \max_{v \in {\cal V}} \frac{1}{Q_V(v)} \right) 2 \cdot 2^{-\beta_1 n}.
    \end{align}
\end{proof}

Note:  Relative entropy can be used to bound total variation via Pinsker's inequality.  With that approach you lose a factor of two in the exponent of decay.  On the other hand, the last steps of the proof can be modified to produce a total variation bound instead of relative entropy.  This direct method keeps the error exponents the same for the total variation case as it is for relative entropy.

\section{Applications}

This stronger version of Wyner's soft-covering lemma has important applications, particularly to information theoretic security.  The main advantage of this lemma comes from the union bound.

The usual random coding argument for information theory uses a randomly generated codebook until the final steps of the achievability proof.  In this final step, it is claimed that there exists a good codebook based on the analysis.  This can be done by analyzing the expected value of the performance for the random ensamble and claiming that at least one codebook is as good as the expected value.  Alternatively, one can make the argument based on the probability that the randomly generated codebook has a good performance.  If that probability is greater than zero, then there is at least one good codebook.  The second approach can be advantageous when performance is not captured by one scalar value that is easily analyzed---for example, if ``good'' performance involves a collection of constraints.

This stronger soft-covering lemma gives a very strong assurance that soft-covering will hold.  Even if the codebook needs to satisfy exponentially many constraints related to soft-covering, the union bound will yield the claim that a codebook exists which satisfies them all simultaneously.  Indeed, if you ran the soft-covering experiment exponentially many times, regardless of how the codebooks are correlated from one experiment to the next, the probability of seeing even one fail is still doubly-exponentially small.

\subsection{Semantic Security}

Wyner's soft-covering lemma has become a standard tool for proving that strong perfect secrecy is achieved in the wiretap channel (see e.g. \cite{bloch-laneman13}).  Coincidentally, Wyner introduced both the idea of soft covering \cite{wyner-common-info} and the wiretap channel \cite{wyner-wiretap} in the same year, but he didn't connect the two together.

According to the usual definition, strong perfect secrecy is achieved if the mutual information (unnormalized) between the message and the eavesdropper's channel output can be made arbitrarily small.

An even stronger notion of near-perfect secrecy is semantic security.  This requires that any two messages cannot be distinguished, usually measured by total variation.  This is not implied by the above strong secrecy because mutual information is an average quantity.  Since there are so many messages, the mutual information can be small even if a few of the messages are perfectly distinguishable.

Semantic security is an operationally relevant metric and widely adopted in cryptography.  In \cite{bellare-tessaro-vardy12} it is shown that semantic security is essentially equivalent to stipulating that the capacity of the channel from the transmitted message to the eavesdropper's observations is negligible, rather than the mutual information with respect to a uniformly distributed message.  They also show that for some binary channels semantic security can be achieved at rates up to Wyner's secrecy capacity.  Note that contrary to the claim in \cite{thangaraj14}, it is not sufficient to analyze the random codebook ensemble for an arbitrary message distribution in order to claim semantic security.  A single codebook must work well for all message distributions.

The soft-covering lemma is used in the proof of the wiretap channel in the following way.  A random codebook is used for communication to the intended receiver; however, two digital messages are concatenated and fed into the encoder (mapped to the codewords):  the actual message to be transmitted; and a random sequence of bits.  This random sequence of bits is what provides the secrecy.  Since the sequence is random, this means that for any individual transmitted message there is a collection of codewords from which one is selected uniformly at random and transmitted.  The soft-covering lemma says that the output at the eavesdropper will look i.i.d. if the size of this set if large enough.  More importantly, this i.i.d. output distribution does not depend on the message that was transmitted.

This argument, using the standard soft-covering lemma (expectation with respect to the codebook), is good enough to claim that the output distribution is close to the i.i.d. distribution on average over the messages.  This can then be used to claim that the mutual information is small.  However, for semantic security, it must be claimed that the output distribution is close the i.i.d. distribution for all messages, and there are exponentially many messages.  Here is where the stronger soft-covering lemma provided in this work is advantageous.  Using the stronger lemma we can claim that a single codebook exists that accomplishes this for every message.

For the single-transmitter wiretap setting, semantic security can be achieved by other means.  The expurgation technique that is used to bound the maximum error probability in channel coding can be used here.  Any offending messages, which do not produce the desired output distribution at the eavesdropper, can be removed from the codebook, and this can be shown to only negligibly reduce the message rate.  However, this expurgation technique will not work in all setting, such as the multiple access wiretap channel.  On the other hand, the proof method involving this stronger soft-covering lemma will work in that setting.  Thus, strong secrecy can be upgraded to semantic security even in situations where vanishing average error probability cannot be upgraded to vanishing maximum error probability.

\subsection{Distributed Channel Synthesis}

In previous work \cite{cuff13}, we characterized the minimum rates of communication and common randomness needed to synthesize a memoryless channel, where the channel inputs are observed at the location of the transmitter, and the channel outputs are produced at the location of the receiver.  This is referred to as distributed channel synthesis.  We say that synthesis is achieved if it is not possible to distinguish the synthetic channel from the genuine memoryless channel that it mimics upon observing the channel inputs and outputs.

The work in \cite{cuff13} only considers the case where the input is a fixed i.i.d. distribution.  A stronger claim would be to say that the synthetic channel cannot be distinguished from the genuine channel even for arbitrary inputs (perhaps with a statistical constraint).\footnote{This stronger claim was shown independently in the work of \cite{bennett-reverseShannon} using an entirely different proof.}  However, the proof in \cite{cuff13} relies heavily on the soft-covering lemma, and the exponential size of even a single type of input sequences made such a claim elusive. A single codebook would need to work well for all input sequences, but the soft-covering lemma only showed that it would work well on average.

With this stronger soft-covering lemma, it may be possible to use the union bound to claim that the soft-covering phenomenon will hold for all of the channel inputs simultaneously.

\subsection{Wiretap Channel II}

The wiretap channel has been studied in other forms aside from the memoryless channel setting.  One such variation, where the eavesdropper gets to make choices about his own channel noise, has been referred to as the Wiretap Channel II \cite{ozarow-wyner84}.  The original formulation was a channel where the eavesdropper is allowed to decide which transmission packets to observed while being limited in quantity.  If the selection of observed packets is an i.i.d. process, then this is the standard wiretap channel setting with an erasure channel to the eavesdropper.  The secrecy capacity of the wiretap channel type II, where the eavesdropper selects the packets to observe, was solved in \cite{ozarow-wyner84} only for the case of a noise-free channel to the legitimate receiver.  Recent work \cite{nafea-yener-wiretap2} investigates the case where the channel to the legitimate receiver is also noisy, for which the secrecy capacity is yet unknown.

The challenge in this setting is that the eavesdropper knows the codebook when it selects the packets to observe.  Therefore, secrecy will only be achieved if it is achieved uniformly for all selections of packets, of which there are exponentially many possibilities.

Using the lemma provided in this work, it can be shown that rates all the way up to the secrecy capacity of the memoryless erasure channel can be achieved even in this more stringent setting.  The codebook construction for the wiretap channel is symmetric in time, so the secrecy analysis, with respect to the random codebook, does not depend on the specific choice of packets observed.  The remaining step that is needed is to show that a single codebook exists which will provide secrecy simultaneously for each one of the exponentially many observation sequences.  This is what the stronger soft-covering lemma provides.

\section*{Acknowledgment}

This work was supported by the National Science Foundation (grant CCF-1350595) and the Air Force Office of Scientific Research (grant FA9550-15-1-0180).

\end{document}